\documentclass[oneside,a4paper,reqno,11pt]{article}
\usepackage{amsmath}

\usepackage{xcolor}
\usepackage{float}
\usepackage{authblk}
\usepackage{amsthm}
\usepackage{amssymb}
\usepackage{bm} 
\usepackage{graphicx}
\usepackage{url}
\usepackage{stackengine}
\usepackage{subcaption}
\usepackage{hyperref}
\numberwithin{equation}{section}
\numberwithin{figure}{section}

\theoremstyle{plain}
\newtheorem{thm}{Theorem}[section]
\newtheorem{lem}[thm]{Lemma}

\newtheorem{pr}[thm]{Proposition}

\newtheorem{theorem}{Theorem}[section]
\theoremstyle{plain}
\newtheorem{conjecture}[theorem]{Conjecture}
\newtheorem{lemma}[theorem]{Lemma}
\newtheorem{proposition}[theorem]{Proposition}

\newtheorem{observation}[theorem]{Observation}
\newtheorem{question}[theorem]{Question}
\newtheorem{problem}[theorem]{Problem}

\theoremstyle{definition}
\newtheorem{definition}[thm]{Definition}

\theoremstyle{remark}

\begin{document}

\title{Condorcet Domains of Degree at most Seven}
\author[1]{Dolica Akello-Egwell}
\author[2]{Charles Leedham-Green}
\author[3]{Alastair Litterick}
\author[4]{Klas Markstr{\"o}m}
\author[5]{S\o ren Riis}
\affil[1,2,5]{Queen Mary University of London}
\affil[3]{University of Essex}
\affil[4]{University of Ume{\aa}}

\maketitle

\begin{abstract}
In this paper we give the first  enumeration of all maximal Condorcet domains on $n\leq 7$ alternatives. This has been accomplished by developing a  new 
algorithm for constructing Condorcet domains, and an implementation of that algorithm which has been run on a supercomputer. 

We follow this up by a survey of the properties of all maximal Condorcet domains up to degree 7,  with respect to many properties studied in the social science and mathematical literature. We resolve several open questions posed by other authors, both by examples from our data and theorems. 

Finally we discuss connections to other domain types such as non-dictatorial domains and  generalisations of single-peaked domains.   All our data are made freely available for other researches via a new website.  

\end{abstract}

\section{Introduction}
Since the seminal treatise on voting by Condorcet \cite{Con}  it has been known that majority voting can lead to collective preferences which are cyclic, and hence does not identify a winner for the election.  Specifically, Condorcet  studied systems where each voter ranks  a list of candidates $A_1,A_2,\ldots, A_n$  and a candidate $A_j$ is declared 
the winner if, for any other candidate $A_i$,  a majority of the voters prefers $A_j$  over $A_i$; here we assume that the number of voters is odd. The candidate $A_j$ is what is 
now called a \emph{Condorcet winner}. However,   Condorcet showed,  \cite{Con}   pages 56 to 61, that there are collections of rankings for three candidates without a Condorcet winner; here 
the pairwise majorities lead to a cyclic ranking of the form  $A_1 < A_2 < A_3 < A_1$. In fact, each candidate loses to one other candidate by a two thirds majority.  This is now often referred  
to as Condorcet's paradox, and the three candidates are said to form a Condorcet cycle.  Ever since Condorcet's result  one has worked to better understand both majority voting  and more general voting systems. 

Going in one direction, looking at which results a vote can actually lead to has been investigated in combinatorics. In order to describe an election result more fully one forms a directed graph  $T$, with the set of candidates as its vertices, a directed edge from $A_i$  to $A_k$ if a majority of the voters rank $A_k$ higher than $A_i$, and no edge if the two alternatives are tied.   Condorcet's paradox demonstrates  that  $T$ may contain directed cycles. McGarvey \cite{McG}  proved that given any specified directed graph $T$, and a sufficient number of voters, there is a set of  preferences for those voters which realise $T$ by majority voting. Results by Erd\H{o}s and Moser \cite{MR168494} and Stearns \cite{MR109087} bounded the number of voters required for tournaments of a given size.  Later Alon \cite{Alon1} also determined how strong the pairwise majorities in such a realisation can be.  

Going in the other direction, Black and  Arrow \cite{Black,Arrow1951} found that if the set of rankings  is restricted in a non-trivial way, either directly  or indirectly, e.g. by voters basing their 
ranking candidates  positions on a common left-right political scale,  there will always be a Condorcet winner, no matter how the votes are distributed over the set of allowed rankings.  This 
motivated the general question: Which sets of rankings always lead to a Condorcet winner? A set of rankings is  now called a \emph{Condorcet domain} if, in a majority vote, it always leads to a linear order on the alternatives, or equivalently $T$ is a transitive tournament. In the 1960's several equivalent characterisations of Condorcet domains were given  by Inada \cite{10.2307/1910176,inada69}, Sen \cite{sen66}, Ward \cite{ward65},  and others.  In particular Ward~\cite{ward65} proved 
that they can be  characterised as exactly those sets which do not contain a copy of Condorcet's original example on three candidates. 

Following these early works  the focus shifted to understanding the possible structure  and sizes of Condorcet domains. Blin \cite{Blin72} gave some early examples with structure different from those by Black and Arrow. Raynaud \cite{ray1981} showed that  if the number of alternatives is at least 4 then there are maximal Condorcet  domains of size just 4. In \cite{Johnson78} Johnson conjectured that the maximum possible size is $2^{n-1}$.    Abello and Johnson \cite{MR763988} investigated the maximum possible size and proved that this is at least $3(2^{n-2})-4$, for $n\geq 5$ candidates, thereby disproving Johnson's conjecture for $n\geq 6$.  They also noted that it was hard to give non-trivial upper bounds for the possible size of a Condorcet domain, and conjectured that the maximum is at most $2^n$.  That conjecture was disproved by Abello in \cite{MR1090284}. Later Fishburn \cite{Fishburn92} showed that the maximum size  grows at least as $c^n$ for some $c>2$, and Raz \cite{Raz} showed that that there is an upper bound of the same form.  By now the maximum possible size has been determined for $n\leq 8$ \cite{LMR23}.

In addition to the size many different structural properties of Condorcet domains have been studied.  Monjardet \cite{M09} surveys many mathematical results on how Condorcet domains relate to the Weak Bruhat order on the set of permutations.   More recent works  have  studied Condorcet domains \cite{SLINKO2019166} with a specific local structure in terms of  Sen's \cite{sen66} value restriction, symmetry properties \cite{KS22}, structure of median graphs \cite{danilov2013} and extensions \cite{PS15} of the original single-peaked property of Arrow and Black.  In \cite{Dittrich2018} Dittrich produced the first full enumeration of all Condorcet domains on $n\leq 5$ alternatives.   A recent survey can be found in \cite{puppe2022maximal}.   Still, much remains unknown  both regarding possible sizes and structures, with open questions motivated both by political science and new applications in computer science.

In this paper we extend the previous results significantly with  the first explicit enumeration of all non-isomorphic Condorcet domains  on $n\leq 7$ alternatives. This has been made  possible by the combination of a new search algorithm developed by us, described in Section \ref{Alg}, and access to a supercomputer.  After presenting basic statistics such as the number of maximal Condorcet domains of given size   we go on to an in-depth investigation of the properties of all Condorcet domains on $n\leq 7$ alternatives.   Here we give data on the number of domains with various well-studied properties,  and we present answers to several open question from the research literature.  Motivated by patterns in our data we present several conjectures on the behaviour of Condorcet domain for large numbers of alternatives.   All our data have been made freely available to download for other researchers via a website which we intend to expand in future works.

\subsection{Outline of the paper}
In Section \ref{defs} we define terminology and discuss various background material.   Section \ref{Alg}  describes our algorithm for generating Condorcet domains. In Section \ref{Data} we discuss of the results of our calculations for degrees  $n\leq 7$, where we have complete enumerations. We also pose a number of questions and give conjectures motivated by the data and our theorems. In Section \ref{other} we discuss connections to other, non-Condorcet, domain types.

\section{Background material and Definitions}\label{defs}
A \emph{Condorcet Domain of degree $n$} is a set of linear orders on a set $X$ of size $n$, satisfying the following definition.
We take $X$ to be the set $\{1,2,\ldots,n\}$, which we write as $X_n$ when we wish to make $n$ explicit.  
\begin{definition}
	A set $S=\{s_1,s_2,\ldots,s_q \}$ of linear orders on $X_n$ is a Condorcet domain if given any three of the 
	linear orders $s_i, s_j, s_k$, and any three  of the elements $a, b, c$ of $X_n$, when we create a 
	table in which each row $r$ is the three elements  ordered according the the $r$:th permutation, that table is not a Latin square. 
\end{definition}
The definition states that the restriction to any three alternatives and any three of the linear orders must not be Condorcet's original example. 
This definition originates with Ward~\cite{ward65} and is one in a long list of equivalent characterisations of Condorcet domains.

It is often convenient to equate a linear order $i_1<i_2<\cdots<i_n$ on $X$ with the permutation $\sigma(j)= i_j$, so a Condorcet Domain
may be regarded as a subset of the symmetric group $S_n$.  So the natural ordering $1<2<\cdots<n$ is equated with the identity map, and the  
reverse ordering $n<n-1<\cdots<1$, which we denote by $u$, is equated with the permutation $(1,n)(2,n-1)(3,n-2)\ldots$, where we write 
permutations as products of disjoint cycles. We refer to an element of  a Condorcet Domain as a permutation or as an ordering, as best fits the context. This switch of point of view is quite common in combinatorial algebra, though as demonstrated in \cite{TOTO20}  the two are essentially different in terms of which properties they can describe in a simple way\footnote{Specifically which properties can be described in first-order logic } and algorithmic complexity.

We will also make use of a second, equivalent, definition of a Condorcet domain, first given by Sen~\cite{sen66}.  Here, a Condorcet Domain $A$ of degree $3$ is defined to be a set of orderings of $X_3$ satisfying one of the 9 \emph{never conditions}, denoted $x\rm{N}i$, meaning that the element $x$ of $X_3$ does not occur in the $i$-th position in any ordering in $A$ Thus $x$N$1$ means that $x$ may never come first, and $x$N$3$ means that $x$ may never come last.  In order to keep our text shorter we will later in the paper use the term \emph{law} 
as a shorter alternative for the tern never condition.    A Condorcet Domain of degree $n>3$ is defined to be a set $A$ of orderings
of $X_n$ with the property that the restriction of $A$ to every subset of $A$ of size 3 is a Condorcet Domain.  In other words, for every triple $\{a,b,c\}$ of elements of $X$ one of the
nine laws $x$N$i$ must be satisfied, where $x\in\{a,b,c\}$; so here $c$N$2$ would mean that $c$ may not come between $a$ and $b$ in any of the orderings in $A$.  It is convenient to take
a Condorcet Domain of degree 2 to be any subset of $S_2$.

By a \emph{Maximal Condorcet Domain} of degree $n$ we mean a Condorcet domain of degree $n$ that is maximal under inclusion among the
set of all Condorcet Domains of degree $n$.  By a \emph{Maximum Condorcet domain} of degree $n$ we mean a Condorcet domain of the largest possible cardinality among those of  degree $n$.  By a  \emph{Unitary Condorcet Domain} we mean one that
contains the natural order, or  the identity permutation depending on how the domain is represented. As we will see in the next subsection every Condorcet domain  is isomorphic to
 some unitary Condorcet domain, so one can usually assume that a domain is unitary without loss of generality; but, as we will see later explicitly making this assumption also leads to various algebraic and algorithmic simplifications.

Henceforth we shall use the acronyms CD, MCD, UCD, and MUCD for the terms Condorcet Domain,
Maximal Condorcet Domain, Unitary Condorcet Domain, and Maximal Unitary Condorcet Domain.

Returning to the case of degree 3 we see that there are nine Maximal Condorcet Domains of degree 3, corresponding to the nine different
laws $x$N$i$.  One checks at once that these nine Maximal Condorcet Domains all contain
exactly four elements, of which, when regarded as permutations, two are odd, and hence are transpositions, 
and two are even, and hence are the identity or a 3-cycle.  Since $S_3$
contains three even and three odd permutations exactly nine subsets of $S_3$ can be 
constructed from two even and two odd permutations, and these are the Maximal Condorcet
Domains of degree 3, described as sets of permutations.  Exactly six of these are unitary,
since the laws 1N1, and 2N2, and 3N3 each rule out a UCD of degree 3.

\subsection{Transformations and isomorphism of Condorcet domains}
Given a permutation $g$ and an integer $i$  we let $ig$  denote $g(i)$, and for a set $A$ of integers, 
$Ag$ is the set obtained my applying $g$ to each element of $A$.

Now, if $A$ is a CD, and $g\in S_n$ is any permutation, then $Ag$ is also a CD; for if $A$ 
satisfies the law $x$N$i$ on a triple $\{a,b,c\}$ for some $x\in\{a,b,c\}$ then $Ag$
satisfies the law $xg$N$i$ on the triple $\{ag,bg,cg\}$.  We say that the CDs $A$ and
$Ag$ are \emph{isomorphic}.  Thus two isomorphic CDs are identical apart from a relabelling of the elements of $X_n$. Every CD $A$ is isomorphic to a UCD, since we can apply $g^{-1}$ to $A$ for any $g\in A$ and obtain an isomorphic UCD.   Similarly we get this lemma, which follows since some element of the first UCD must be mapped to the identity order in the second UCD. 
\begin{lemma}
	If two UCDs $A$ and $B$ are isomorphic then $Ag^{-1}=B$ for some $g$ in $A$. 
\end{lemma}
The lemma leads to the following observation.
\begin{proposition}\label{isoprop}
	Isomorphism between two CDs of equal size can be tested in time which is polynomial in the size of the domain and $n$.
\end{proposition}
\begin{proof}
	Let $A$ and $B$ be two CDs. Form $A_1=Ag^{-1}$ for some $g\in A$ and $B_1=Bh^{-1}$ for some $h\in B$. Clearly $A_1$ and $B_1$ are unitary 
	and isomorphic to $A$ and $B$ respectively, and $A$ is isomorphic to $B$ if and only if $A_1$ and $B_1$ are isomorphic.
	
	In order to test isomorphism of $A_1$ and $B_1$ we simply need to check if $A_1g^{-1}=B_1$ for any $g\in B_1$. This requires at most $|B_1|$ tests, and 
	each test can be done in time  $O(|A_1|n)$ using the Radix sort-algorithm, assuming that the permutations are stored as strings of length $n$.  
\end{proof}
The run time given by the simple algorithm described here is not optimised for small domain sizes.  For small domains the radix-sort step could be replaced by e.g. insertion sort.

\begin{definition}
	The \emph{core} of a UCD $A$ provedthe set of permutations $g\in A$ such that $Ag=A$.  
\end{definition}
Since $A$ is unitary the core of $A$ is a group.  We will study the properties of the core and other symmetries of a UCD, both for small $n$ and in general, in a later paper.

When we speak of an isomorphism class of UCDs we mean the set of UCDs in an isomorphism class of CDs.
So if $A$ is a UCD of size $m$, with core of size $k$, then $k$ divides $m$, and the isomorphism 
class of $A$, as a UCD, is of size $m/k$.

\begin{definition}
	The \emph{dual} of a CD $A$ is the CD obtained by reversing each linear order in $A$.
\end{definition}
Equivalently the dual is given by $uA$, when $A$ is viewed as a set of permutations.
Note that if $A$ satisfies the law $x$N$i$ on some triple then $uA$ satisfies the law
$x$N$(4-i)$ on the same triple.  Thus $A^u=uAu$ is also a CD, and if $A$ is a UCD
then so is $A^u$.

\begin{lem}
	For every $n>1$ the map $A\mapsto A^u$ permutes the set of isomorphism classes UCDs  of degree $n$,
\end{lem}

\begin{proof}
	Let $A$ and $B=Ag^{-1}$ be UCDs of degree $n$, where $g\in A$.  Then $B^u=(Ag^{-1})^u=
	A^u(g^{-1})^u$.  But $(g^{-1})^u=(g^u)^{-1}$, and $g^u\in A^u$; so $B^u$ is isomorphic to $A^u$, as required.
\end{proof}

\begin{definition}
If $E$ is an isomorphism  class of UCDs such that $E^u=E$ we say that $E$ is \emph{reflexive}.
If this is not the case we say that $E$ and $E^u$ are \emph{twinned}.  

If $A$ and $B$ are UCDs that  are isomorphic, or in twinned isomorphism classes, we say
that $A$ and $B$ are \emph{isometric}. This is also known as being flip-isomorphic. 
\end{definition}

\subsection{The weak Bruhat Order and Condorcet domains as posets}
The weak Bruhat order is a partial order on the set of permutations $S_n$, and hence also on the the set of linear orders. A number of results on CDs have been proved using the structure of this linear order and we shall classify CDs according to some such properties.

Given a linear order $\sigma$, here seen as a permutation,  an  \emph{inversion} is a pair $i<j$ such that $\sigma(i)>\sigma(j)$ and we let $Inv(\sigma)$ denote the set of all inversions for $\sigma$.   The weak order is defined by saying that $\sigma_1\leq \sigma_2$ if $Inv(\sigma_1)\subseteq Inv(\sigma_2)$.  We say that $\sigma_2$ covers $\sigma_1$ if $\sigma_1 \leq \sigma_3 \leq \sigma_2$ implies that $\sigma_3$ is equal to one of $\sigma_1$ and $\sigma_2$.    By the \emph{Hasse diagram} one means the directed graph with vertex set $S_n$ and a directed edge from $\sigma_1$ to  $\sigma_2$ if $\sigma_2$ covers $\sigma_1$.

The weak order turns the  set of linear orders, or equivalently the symmetric group $S_n$, into a partially ordered set  known as the \emph{permutohedron}.  Since a CD $A$ can be viewed as a subset of the permutohedron   we also get an induced partial order on the elements of $A$.   Note that the dual CD for  $A$  induces the dual, in the poset sense, partial order of $A$.

It was noted already by Blin~\cite{Blin72} that a maximal chain in the permutohedron  is a Condorcet domain.

\begin{definition}
	A CD $A$ is Bruhat-self-dual if it is isomorphic to the dual of $A$.
\end{definition}
Note that in terms of posets this means that  $A$, as a poset, is isomorphic to the dual poset of $A$.

\begin{definition}
	A CD $A$ is connected  if for any two $a, b \in A$ there exists  sequence $a=\sigma_1,\sigma_2,\ldots,\sigma_k=b$, with each $\sigma_i\in A$, such that either $\sigma_i$ covers $\sigma_{i+1}$, or $\sigma_i+1$ covers $\sigma_{i}$ in the permutohedron.
\end{definition}
This definition states that $A$ induces a weakly connected subgraph in the Hasse diagram of the permutohedron.

\subsection{Bounds for the size of a MCD}
Perhaps the main focus of research in this area has been the attempt to find reasonable
 bounds for $F(n)$,  a function introduced by Fishburn \cite{Fis96}  to denote the maximum size of an MCD of given degree $n$.  A lower bound for $F(n)$ is obtained by
two recipes (the alternating scheme and replacement schemes), which we describe below.

The \emph{alternating scheme}, discovered by P.C. Fishburn see \cite{Fis96} and \cite{Fishburn97},  gives rise to the largest possible MCDs of degree up to 7, as we later prove
with our calculations, but the replacement schemes can do better in degrees greater than 15,
and perhaps for some smaller degrees.  There are two isomorphic alternating
schemes ${\mathcal A}_n$ and ${\mathcal B}_n$ of degree $n$; 
${\mathcal A}_n$ is defined by the following laws.  For every triple $a<b<c$ the law $b$N$1$ is imposed
if $b$ is even, and the law $b$N$3$ is imposed if $b$ is odd.  
Similarly ${\mathcal B}_n=u{\mathcal A}_n$ is defined 
by the laws $b$N$3$ if $b$ is even, and $b$N$1$ if $b$ is odd.  Clearly ${\mathcal A}_n$ and
${\mathcal B}_n$ are UCDs.  Galambos and Reiner prove in \cite{GR} that 
$\vert {\mathcal A}_n\vert=2^{n-3}(n+3)-{n-2 \choose n/2-1}(n-3/2)$
if $n>3$ is even, and $\vert {\mathcal A}_n\vert=2^{n-3}(n+3)-{n-1 \choose (n-1)/2}(n-1)/2$ if $n>2$ is odd, and also prove that these UCDs are maximal.

Fishburn's second method for constructing CDs is the \emph{replacement scheme}, defined thus.
Let $A$ and $B$ be CDs on the sets $Y={1,2,\ldots,k+1}$ and $Z={k+1,k+2,\ldots,k+l}$.
Then a CD $C$ on $X_{k+l}$ is obtained by taking all the elements of $Y$, as orderings, and
replacing all occurrences of $k+1$ by elements of $Z$.  So $C$ is a CD on $X_{k+l}$, and
$|C|=|A| |B|$.  Here $k$ and $l$ may be equal to 2, and one sees at once that the CD of
degree 3 defined by the law 1N2 is a replacement scheme, with $k=l=2$.  Clearly if
$A$ and $B$ are unitary then so is $C$, and if $A$ and $B$ are maximal then so is $C$.

Ran Raz proves in \cite{Raz} that there is an upper bound for $F(n)$ of the form $c^n$ for some universal constant $c$.   His proof
covers a wider class of sets of linear orders than Condorcet domains, but looking at his parameters in the case of
alternating schemes it is clear that his argument will not yield a realistic value for $c$ in the
case of CDs.  Fishburn's schemes imply~\cite{Fishburn97} that $c>2.17$  and Conjecture 3 of that paper would imply that $c\leq 3$.

\subsection{Closed CDs and sets of laws}\label{galois}
As a final general remark, there is a Galois type correspondence between subsets of $S_n$, or
\emph{permutation sets}, and sets of  laws, in which a permutation set corresponds
to the set of laws that are obeyed by every permutation in the set, and a set of laws
corresponds to the set of permutation sets that satisfy these laws.  This gives rise to the concepts
of a \emph{closed} set of laws, which is a set $L$ of laws that contains all laws that are consequences 
of laws in $L$, and of a \emph{closed} permutation set, which is a permutation set $A$ that contains
all permutations that satisfy all the laws satisfied by all the elements of $A$.    Clearly all MCDs are closed, also the replacement scheme obtained from two closed permutation sets is clearly closed.

Call the set of elements of $S_n$ that satisfy a given  law a \emph{principal}
closed permutation set.  These all have cardinality $2n!/3$, and the closed permutation sets are 
precisely the intersections of sets of principal permutation sets.
In our algorithm to construct all MUCDs of a given degree we only consider closed permutation sets,
and we are concerned with the closure of sets of laws.  However, we do not have a good theoretical grip on
these concepts.  The only algorithm that we use for determining the closure of a set of laws
is to go back to the definition, construct the set of permutations that obey these laws, and see what
further laws these permutations all obey, and similarly for the closure of a permutation set.
It may be that the lack of a theoretical insight into the nature of closure is related to the difficulty
in proving theorems about CDs, and in particular about MCDs.  For example, this prevents us from
obtaining a good complexity analysis of our algorithm.

\section{The Generation  Algorithm}\label{Alg}
Next we will describe our algorithm for generating all MUCDs of a given degree  $n$.   We have implemented this algorithm in C, both in a serial version which is 
sufficient for degree $n\leq 6$, and a parallelized version which was used for $n=7$.

Our first step is to arrange the ${n\choose 3}$ triples of integers in $X_n$ in some fixed order,
and to construct and store all the principle closed subsets of $S_n$, as defined in section \ref{galois}. We also fix an ordering of the set of  laws.

To a first approximation the algorithm operates in the \emph{full Condorcet tree},  which is a homogeneous rooted tree of depth ${n \choose 3}$, where every non-leaf has six descendants and every edge is labelled by a  law.
Each vertex of the tree will be assigned a closed permutation set. For the root vertex this is the set of all $n!$ permutations of $X_n$, and for lower vertices the set is constructed recursively from the set on its parent in the following manner.    Every edge joining a vertex of depth $t$ to a vertex of depth $t+1$ is associated with one of the six laws that may be applied to
the $t$-th triple, the numbering being organised in such a way that the root is associated with the
first triple.  Thus each edge is associated with a unique  law. The permutation set associated with the vertex of depth $t+1$ is inductively defined
as the intersection of the permutation set associated with the vertex of depth $t$ with the principle closed permutation set that is associated with the edge in question.  Clearly 
every MUCD will appear at least once as the permutation set associated with some leaf of this tree.  However, for degree 6 we have a tree with $6^{20}$ leaves, making the computation  infeasible. Additionally, using this tree is very inefficient from a computational point of view since it actually contains all Condorcet domains, both maximal and non-maximal as well as all members of every isomorphism class of MUCDs, whereas we only need one such member. 

In constructing our algorithm we  restrict our search to a sub-tree of the full Condorcet tree such that only maximal domains are constructed, and  each MUCD is constructed exactly once.. Doing this will lead to a tree in which every retained internal vertex of the Condorcet tree has 0, 1 or 6 descendants, depending on whether the permutation set at that vertex is non-maximal/redundant,  has an implied law on the current triple, or is unrestricted by our application of laws to earlier triples. 

\subsection{Implied laws and redundancy}
The first restriction on our search comes from implied laws.   When we have applied laws to  a sequence of triples  it can happen that they imply a law on some triple. The latter triple can either be one of the triples we have already visited when applying laws or a triple we have not yet visited. Each of these two  cases lead to a reduction  of our search.

Let us first note that we can view a sequence of triples, coming in the order we have specified on the set of triples, together with the applied laws  as a string over an alphabet of size 6, the number of laws.   Since  we have also defined an order on the set of laws we can sort any set of such string using their lexicographic order.    Now in each isomorphism  class of of MUCDs  we only need to keep the one which is lexicographically  largest, if our aim is to generate representatives for each isomorphism class.  In our algorithm we do not go that far but we only keep vertices  which correspond to a string which is obviously not lexicographically maximal.    This is done as follows.

Whenever a law is applied to a vertex $v$   we  compute the set of laws satisfied at each triple along the path from $v$ to the root.  At each such triple we know which law was applied and which laws are now implied.  If one of the implied laws precedes the applied law, in our order on the set of laws, then the search at vertex $v$ is abandoned, leading to 0 descendants. The reason is that there will be another path in the tree where the role of the implied and applied are switched, hence making that sequence lexicographically larger while leading to the same permutation set. 

Our second case is that where we reach a new vertex $v$ and find that this triple already has an implied law.  In this case we only generate the single descendant which corresponds to the implied law. Any other descendant of $v$ will hold a permutation set  which is a strict subset of the one we actually generate and hence not maximal.  At late stages in the search we often find that all remaining triples have implied laws and hence do not lead to branching of the search tree.

We call the  tree resulting from these restrictions the \emph{reduced Condorcet tree}.

\subsection{Maximality}
While the previous restrictions lead to a much smaller search tree they still leave many non-maximal UCDs in the tree.  Our next step is to restrict the search to only permutation sets which can lead to a maximal UCD, and only MUCDs as the final leaves at depth  ${n\choose3}$. 

For any triple $t$, let us define a $t$-UCD to be a permutation set that satisfies some law for every triple $s<t$, and define a $t$-MUCD to be a maximal $t$-UCD 
permutation set, with respect to inclusion, so that every $t$-MUCD is a closed permutation set.  If $t={n\choose3}$ then a $t$-UCD is a UCD, and
a $t$-MUCD is an MUCD.

We can now formulate the following proposition.
\begin{pr}
	Let $n>3$, and let $1\le t \le {n\choose3}$.  Then every $t$-{\rm MUCD} $A$ is associated with a vertex in the reduced 
	Condorcet tree whose parent is associated with a $(t-1)$-{\rm MUCD}.
\end{pr}
\begin{proof}
	Since every closed $t$-UCD occurs as the permutation set associated with a vertex in the full
	Condorcet tree it follows that every $t$-MUCD occurs as the permutation set
	associated with some vertex in the reduced Condorcet tree.
	Let $A$, as in the proposition, be associated with a vertex $V$ in the reduced Condorcet tree,
	so that $V$ is a child of a vertex $W$ with corresponding triple $s<t$.  The edge joining $V$ to $W$ is
	labelled by a law applied to the triple $t-1$.  Let $B$ be the $(t-1)$-UCD associated with $W$, and let
	$L$ be the law that labels the edge joining $V$ to $W$.  If $B$ is not a $(t-1)$-MUCD then there
	is a vertex $W'$ that is associated with a polynomial set $B'$ that contains $B$ and that is
	a $(t-1)$-MUCD.  The child of $W'$ defined by the law $L$ is associated with a $t$-UCD $A'$ that
	contains $A$.  From the maximality of $A$ it follows that $A=A'$, and the proposition is proved.
\end{proof}
Using this proposition leads to a considerable improvement in the performance of the algorithm.
We only process the subtree of the reduced Condorcet tree whose associated permutation sets
are $t$-MUCDs for the appropriate $t$.  

It remains to describe how we decide if a permutation set $A$ satisfies this condition. This is achieved via the following lemma. 
\begin{lemma}
	Let $A$ be the permutation set associated with a vertex $v$ at depth $t$,   For every $s<t$ let  $L_s$ be the set  of laws that $A$ satisfies on the 
	triple $s$, let $M_s$  be the union of the corresponding principal closed permutation sets, and let $B$  be the intersection of the sets $M_s$.  
 	Then $A$ is a $t$-MUCD if and only if $A=B$.
\end{lemma}
\begin{proof}
	Since in any case $A$ is contained in $B$, the condition $A=B$ is equivalent to the condition that
	$B$ is contained in $A$.  Suppose not, and let $b\in B\setminus A$.  Then clearly $A\cup\{b\}$ is
	a $t$-UCD that properly contains $A$.  Conversely, if $A$ is not a $t$-MUCD then there is some
	$b\in A$ such that $A\cup\{b\}$ is a $t$-UCD.  But then $b\in B$, and this completes the proof.
\end{proof}

\subsection{Final reduction, parallelisation, and implementation}
The algorithm described produces a list which contains at least one representative for  each isomorphism class of non-isomorphic MUCDs.   However there are still repeated members from some classes. In order to produce the list of all non-isomorphic MUCDs  we compute the isomorphism class of each leaf, using the  observations before Proposition \ref{isoprop}, and outputting the lexicographically maximal member of each such class.  The list of such CDs is then sorted and duplicates removed in order to produce our final list. The isomorphism  reduction was done by a separate program after the search, and was later also done with the independently coded CDL library \cite{zhou2023cdl} as an independent verification. 

The parallel version of this algorithm first finds all vertices at a user specified distance from the root of the search tree and outputs them into a file.  Next, independent copies of the 
program completes the search of the subtrees rooted at each of the  vertices in the file. Finally the outputs from these searches are merged in the same way as for the serial version. 

It remains to say something about the technical details of our implementation.  The elements of $S_n$ are enumerated, so that a subset of $S_n$ may be represented as a bit-string of length $n!$.  The principal closed permutation sets are computed as bit-strings in a pre-processing stage, and all further computations with sets of permutations are carried out using bit operations.  

The correctness of our program was tested against full enumerations of MUCDs for small $n$ generated by other programs, using brute force enumeration.

\section{The MUCDS of degree at most 7 and their properties.}\label{Data}
Using our algorithm we have made a complete enumeration of all MUCDs  of degree $n\leq 7$.  The total numbers for $n$ from 3 to 7 are 3, 31, 1362, 256895, 171870480. Reducing further to flip-isomorphism classes we get $2, 18, 688, 128558, 85935807$. Here the first four numbers in both cases agree with published results and the final one is new.    
The MUCDs are available for download \cite{Web1}.

In the next subsections we will discuss our computational analysis of these MUCDs and  their properties. We will provide counts for the number of MUCDs with certain well studied properties and the distribution of properties which have a range of values. We  also test several conjectures from the existing literature and report on those results.

\subsection{The sizes and numbers of MUCDs}
In Tables  \ref{tab45}, \ref{tab6}, \ref{tab71} and \ref{tab72}   we display the number of MUCDs  of degree 4 to 7  listed according to various properties.  In each table the column labelled Total  gives the number of MUCDs with the size stated in the previous column.

Using our results we can settle a conjecture whose status has been uncertain for some time.  In~\cite{ Fishburn97} Fishburn conjectured that for $n=6,7$ a CD is maximum  if and only if it is isomorphic to those constructed by his alternating scheme. He also proved that the same statement is true for  $n=4, 5$.  In~\cite{Fishburn2002} he provided a long, and according to himself partial, proof for the case $n=6$. His caveat was not due to any uncertainty in the proof, but rather since the considerable length of the proof made him leave many details out of the published version.   In \cite{GR}, Section 3.2,  Galambos and Reiner stated that they verified the conjecture for $n=7$, but gave no details regarding how this was done.   The lack of a published proof led the recent survey~\cite{surv22} to list even the maximum size for $n=7$ as unknown. Using  our data we now have a computational verification of Fishburn's proofs for $n=4,5,6$ and a proof of his conjecture for $n=7$.
\begin{theorem}
	For $n=4,\ldots,7$ every maximum CD is isomorphic to a MUCD constructed by Fishburn's alternating scheme. In particular, the maximum size of a CD for $n=7$ is 100. 
\end{theorem}

We also note, using \cite{LMR23}, that for $n\leq 8$   the maximum CDs have size  $\left\lceil 4\times 5^{\frac{n-3}{2}}\right\rceil$.   That such a simple form will continue to hold might be too much to hope for  but the growth rate is compatible with known data and bounds.
\begin{problem} 
	Let $F(n)$ denote the size of a maximum CD on $n$ alternatives.  Prove or disprove that $$\lim_{n\rightarrow\infty} \frac{\ln(F(n))}{n}=\sqrt{5}$$
\end{problem}

Next, as we can see the total number sequence for a fixed degree is not unimodal, though roughly so. The sequence achieves its largest values at slightly more than half the size of the maximum MUCD for each degree, but it is strongly affected by parity and divisibility by larger powers of 2.   In Figure \ref{fig:si7} we display the size counts for $n=7$.
\begin{figure}
	\centering
	 \includegraphics[width=0.5\textwidth]{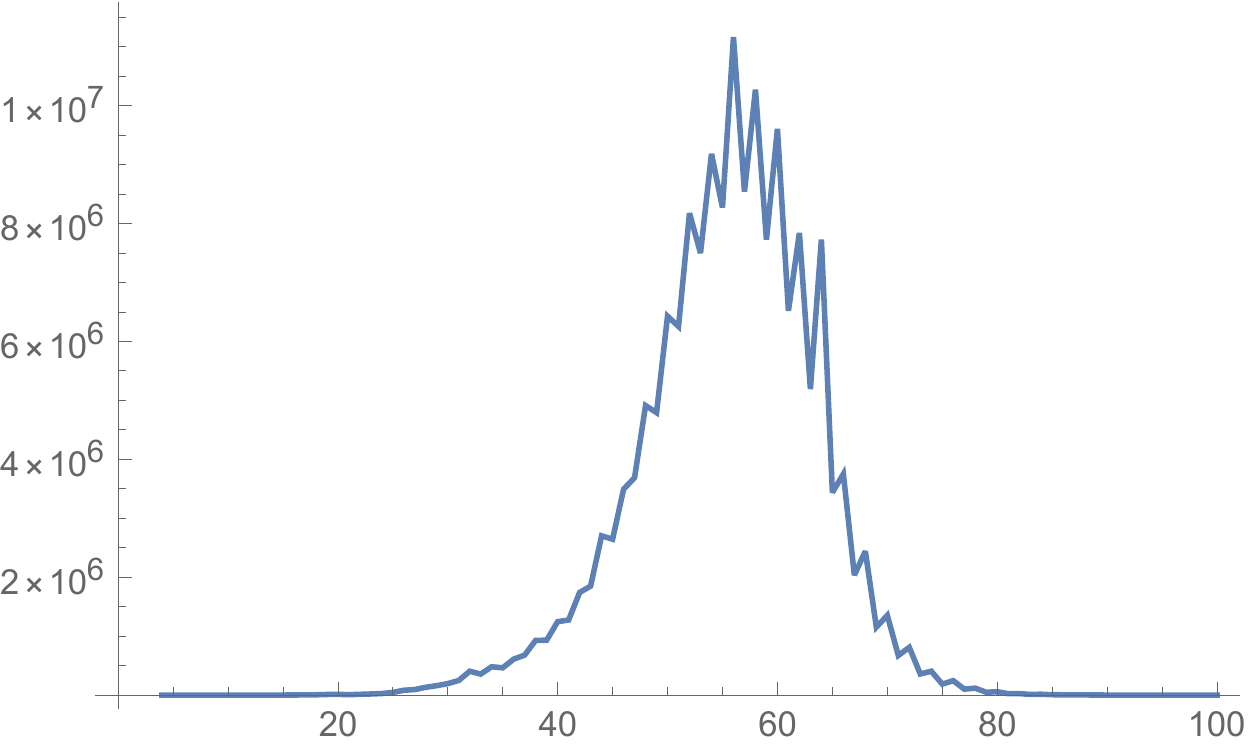}
	\caption{The number of MUCD classes as  function of domains size for $n=7$}\label{fig:si7}
\end{figure}
We can create a natural notion of a random MUCD by giving each isomorphism class equal probability and taking a random member of the chosen isomorphism class. 
The expected size of a MUCD under this distribution is for each small degree lower than $2^{n-1}$ but can be very well fitted to an exponential function.
\begin{conjecture}
	Let $Z_n$ be a random MUCD  then $\log( \mathbb{E}(|Z|))\sim n\log(q)$, for some constant $1<q<2$.
\end{conjecture}
Fitting an exponential function to the four, admittedly few, values we have for $\mathbb{E}(|Z|))$ gives a very good fit to $0.59163\times 1.91324^n$.  Fitting the variance  also give a good fit to an exponential growth of 4.663. The third moment is negative and gives a negative skewness which is growing in magnitude for our range of $n$.  With all of this in mind it seems likely that a the size distribution converges after a proper normalisation but it is not clear what the asymptotic form will be.  
\begin{question}
	Let $M_n$ and $\sigma_n$ be the mean and standard deviation of $|Z_n|$ and define $Y_n=\frac{|Z_n|-M_n}{\sigma_n}$. 
	
	Does $Y_n$ converge in distribution as $n\rightarrow \infty$? If so, what is the asymptotic distribution?   
\end{question}

\subsection{The  structure of MUCDs}
The first structural  property which we will look at is whether or not a MUCD can be built from CDs of lower degree.
\begin{definition}
	Given a MUCD $C$  on a base set $A$  we say that $C$ is \emph{reducible} if there exists a proper subset $B\subset A$, of size at least 2, such that the elements of $B$ 
	are consecutive in each of the linear orders in $C$. If $C$ is not reducible we say that it is irreducible.
\end{definition}
The motivation for this definition is that a reducible MUCD  can be built from two CDs, $C_1 $ on a set $A'$ of size $1+|A\setminus B]$   and  $C_2$ on  $B$, using a slight generalisation of Fishburn's replacement scheme. There we pick some element of $A'$  and then replace that element in every member of $C_1$ with a permutation from $C_2$.   In the column labelled Reducible we display the number of MUCDs of each size which are reducible.  Obviously reducibility is strongly affected by the factorisation of the size, since the size of a reducible MUCD is the product of the size of the factor CDs $C_1$ and $C_2$, each of which must be maximal.    Even though the number of reducible MUCDs increase with the degree  we nonetheless expect them to asymptotically be outnumbered by the irreducible ones.
\begin{conjecture}
	MUCDs are asymptotically almost surely irreducible\footnote{That a property holds asymptotically almost surely, abbreviated a.a.s.,  means that as $n$ goes to infinity the proportion of objects with the property goes to 1.}.
\end{conjecture}

Next we  see that for each degree we find several MUCDs of size $4$. The first such examples were found by Raynaud \cite{ray1981}   and Danilov and Koshevoy \cite{danilov2013} proved  that these exist for all degrees. These domains can be used to construct MUCDs for larger powers of 2 as well  and we may ask for which fixed sizes we can find a MUCD for infinitely many, or all sufficiently large, degrees.
\begin{question} 
	Are there infinitely many degrees for which a MUCD of size 9 exists?  For which sizes $t$ do there exists MUCDs for infinitely many degrees $n$?
\end{question}

We now look at the  set of laws, or never conditions, a MUCD satisfies.    A particularly nice subfamily of the MUCDs are those which satisfy exactly one law on each triple of alternatives. These MUCDs were named \emph{copious} by Slinko \cite{SLINKO2019166}, the name alluding to the fact that a copious CD gives the maximum possible 4 orders when restricted to any triple of alternatives.   In the column labelled Cop  we show the number of copious MUCDs of each size.  For $n\geq 5$ we find examples for CDs which are not copious. For $n=5$ the restriction to a triple either has size 3 or 4. For $n=6$ all MUCDs with size 9 or less  have restrictions of size 2 or 4, thus being even further from being copious.  We also see that for $n\leq 7$  MUCDs which are close to the maximum size are always copious, and that for most of the range of sizes they make up the majority of all MUCDs.  However, in order to be copious the restriction of a MUCD to a subset of the alternatives must be copious as well. That requirement could make copious MUCD less common for lager $n$. 
\begin{question} 
	What is the minimum size of a copious MUCD of degree $n$?   Are asymptotically almost all MUCDs  not copious?
\end{question}

In \cite{KS22} Karpov and Slinko used the term \emph{ample} to denote those CDs which, whenever restricted to two alternatives  give both of the possible orderings for those alternatives and noted that a copious MUCD is ample. They asked if all MUCDs are ample and we can answer this question negatively: 
\begin{observation}
	The smallest non-ample MUCD  has degree 5 and size 12 
\end{observation}
\begin{figure}[h]
\centering
\begin{tabular}{cccccccccccc}
1 & 1 & 1 & 1 & 1 & 1 & 1 & 1 & 4 & 4 & 4 & 4 \\
2 & 2 & 2 & 3 & 3 & 3 & 4 & 4 & 1 & 1 & 2 & 3 \\
3 & 3 & 5 & 2 & 2 & 5 & 2 & 3 & 2 & 3 & 1 & 1 \\
4 & 5 & 3 & 4 & 5 & 2 & 3 & 2 & 3 & 2 & 3 & 2 \\
5 & 4 & 4 & 5 & 4 & 4 & 5 & 5 & 5 & 5 & 5 & 5 \\
\end{tabular}
\caption{The smallest non-ample MUCD}\label{nonamp}
\end{figure}
The number of non-ample  MUCDS of each size is displayed in the column labelled Non-amp. For $n=5$  there are only  3 non-ample MUCDs, but as the degree goes up they become more common.  Note that for degree $n=7$ all MUCDs of size 9 are non-ample.  We also find surprisingly large examples of non-ample MUCDs for $n=6$  with size 40,   and $n=7$ with size 93.
\begin{question} 
	Is the maximum size of a non-ample MUCD $o(F(n))$?
\end{question}

Being non-ample is not the only deviation from what one might at a first glance expect a maximal UCD to look like.   Let us say that a UCD  $C$ is \emph{fixing} if there exists a value from the base set which has the same position in every order in $C$.  It is clear that if we take a Condorcet domain and insert a new alternative at a fixed position in every linear order we will get a new Condorcet domain, of the same size and degree one larger.  One would typically not expect such a CD to be maximal, however it turns out that it is possible to construct MUCDs in this way.   
\begin{observation}
	The smallest fixing MUCD  has degree 5 and size 4. 
\end{observation}
\begin{figure}[h]
\centering
\begin{tabular}{cccc}
1 & 2 & 4 & 5 \\
2 & 5 & 1 & 4 \\
3 & 3 & 3 & 3 \\
4 & 1 & 5 & 2 \\
5 & 4 & 2 & 1 \\
\end{tabular}
\caption{A fixing MUCD of order 5 and size 4.}
\end{figure}
For degree 5 there is a unique fixing MUCD, and for degree there are 2, both with size 8. For degree 7, there are 6 with size 4, 3 of size 8, 4 of size 13 and 133 of size 16.
\begin{question}
	How large can a fixing MUCD of degree $n$ be? Is there a characterisation of the MUCDs which have an extension to a fixing MUCD with one more alternative?
\end{question}

\begin{table}
\begin{tabular}{|l|l|l|l|l|l|l|l|l|l|l|}
\hline
Degree &Size  & Total  & Connected  & Normal & \stackanchor{Self-dual}{(Symmetric)}  & Non-ample & Reducible &  Cop\\
 & &  &   &   &    &   && \\
\hline
4 & 4 & 1 &   &  1   & 1\hfill (1) &  & & 1 \\
\hline
4 & 7 & 4  & 2  &  4 &   &  &  & 4 \\
\hline
4 & 8 & 25 &  7  &  16  & 3\hfill (2) &  & 8& 25 \\
\hline
4 & 9 & 1  & 1  & 1  &  1 &   &  & 1\\
\hline
\hline

&&&&&&\\
\hline
5 & 4 & 2  &   &  2   & 2 \hfill(2)  &     &  &  \\
\hline
5 & 8 & 12 &   &  8  & 2\hfill (2) &     & 2 & 12\\
\hline
5 & 11& 28 & 2  &  18 &   &  &  & 26\\
\hline
5 & 12 & 41 & 16 &  32  & 1  &  1 &  & 36\\
\hline
5 & 13 & 52  & 2 & 32   &   &  &  & 44 \\
\hline
5 & 14 & 279 & 26  & 118   & 1    & 1 & 20 & 236\\
\hline
5 & 15 & 212 & 42  & 58 &    & &   & 208\\
\hline
5 & 16 & 573   & 57 & 141  & 7 \hfill(3)   &  1   & 100 & 572\\
\hline
5 & 17 & 106   & 20  & 34  &    &  & &106\\
\hline
5 & 18 & 43   &  6  & 19  &  1 &  & 5 & 43\\
\hline
5 & 19 & 12  & 8  &  6  &   &  & & 12 \\
\hline
5 & 20 & 2  & 2  & 2  &  & & & 2\\
\hline
\end{tabular}
\caption{MUCDs of degree 4 and 5}\label{tab45}
\end{table}

\begin{table}
\begin{tabular}{|l|l|l|l|l|l|l|l|}
\hline
Size  & Total & Connected  &Normal&  \stackanchor{Self-dual}{(Symmetric)} & Non-ample &  Reducible &  Cop\\
   &   &  &  &   &  &   &\\
\hline
 4 & 8  &    &  8  & 8\hfill (8)  &   & &  \\
\hline
 8 & 11  &    &  7 & 7\hfill (7) & 4    &7 & \\
\hline
 9& 26  &   & 18  &   && & \\
\hline
 10 & 46 &   &  28 &  &&  & \\
\hline
 11& 8  &    & 6   &   &  & & \\
\hline
 12& 11  &   &  4  & 1  & 7& & \\
\hline
 13& 106 &   & 38  &    &4&  &90\\
\hline
 14& 80  &    &  32 &   &2 & & 76\\
\hline
 15& 66  &     &  34 &   &2 & & 54\\
\hline
 16& 1036  &  2  &  246  & 8\hfill (6)   &  6    & 62 & 970\\
\hline
 17& 808  & 12  & 244  &    &  & &642\\
\hline
 18& 808  & 14  & 280  &    &16& & 600\\
\hline
 19& 1399  &  76 & 537  &  3   &40& & 1125\\
\hline
 20& 1734 & 144  & 664  &   4  &45& & 1333\\
\hline
 21& 2156  & 124  & 708  &  2    &114& & 1486\\
\hline
 22& 5072 & 100  & 1194  &     &164& 168 & 3876\\
\hline
 23& 4986  & 114  & 1378  &     &  108& &  3372\\
\hline
 24& 8617 &  246  &  1850  &  9  &207&237 & 5964\\
\hline
 25& 9892  & 240  & 1624  &  2   &156& & 7014\\
\hline
 26& 16629 & 491   &  2502  &  5   &164&312 & 11345\\
\hline
 27& 17137 & 739  & 1756  &  3    &138 & & 12269\\
\hline
 28& 32708  & 883   & 3100   &  16   &281& 1604 & 27013\\
\hline
 29& 25453  & 1176 & 1760  &  5   &168 & & 21909\\
\hline
 30& 31310  & 1420  & 2289  &  6    &188&1272 & 28820\\
\hline
 31& 22543 & 1099  &  1381 &  7  &114& & 21159\\
\hline
 32& 38894 & 1022  & 2195   & 46\hfill (6)    &  307 & 3127  &  37885\\
\hline
 33& 12168 & 548   &  821   &  24   & 84 & & 11722\\
\hline
 34& 11554  & 490  &  1075  & 10    &70 & 636 & 11332\\
\hline
 35& 4635  & 332  & 532  &  7  &38 & & 4573\\
\hline
 36& 3720  & 232  & 458   &  22  & 92  & & 3620\\
\hline
 37& 1297  & 144  &  177  & 11  & 8 & & 1283\\
\hline
 38& 1300 & 114 &  284  & 2   &18 & 72 & 1282\\
\hline
 39& 366  & 79  &  70 & 2  & &  & 366\\
\hline
 40& 192 & 35  & 41  & 2  &5 & 8 & 187\\
\hline 
 41& 50   & 22  &  16 &  && &   50\\
\hline
 42& 57   & 31  & 15  & 7 &&  &  57\\
\hline 
 43& 7   & 5  & 2   & 1  &&  & 7\\
\hline
 44&  4  & 4  &  2 &   && & 4\\
\hline
45&  1  & 1  &  1  & 1  && & 1\\
\hline

\end{tabular}
\caption{MUCDs of degree 6 }\label{tab6}
\end{table}

\begin{table}
{\small 
\begin{tabular}{|l|l|l|l|l|l|l|l|l|}
\hline
Size  & Total & Connected  & Normal &  \stackanchor{Self-dual}{(Symmetric)} & Non-ample & Reducible &  Cop \\
\hline
4       &  46  &   & 46   & 46\hfill (46) &    & &    \\
\hline
8       & 44   &   & 44  & 44\hfill (44) &   & 44 &    \\
\hline
9       & 24   &   &   &   & 24 & & \\
\hline
10      & 270   &    & 186  &  &  10 & & \\
\hline
11       & 188   &   & 120  &    & 2 & & \\
\hline
12       & 147   &   &  84 & 1  & 3  & & \\
\hline
13       & 176  &   & 72  &    & 4 & &  \\
\hline
14       & 284  &    & 110  &  4  & 60 & & \\
\hline
15       & 548  &    &  188 &    & 112 & & \\
\hline
16       & 1626  &   &  452  &  26\hfill (24)   &  228  & 60 &   59\\
\hline
17       & 2178   &   &  490 &     &  358 & & 30\\
\hline
18       & 4435   &   & 794  &  1   & 283  & 182 & 14\\
\hline
19       & 9994   &   &  1248 &     & 358 & & 1528\\
\hline
20       & 11864   &    &  1544  & 4    &502  & 322& 1630\\
\hline
21       & 7040   &   &  1274 &     & 702 & & 1338\\
\hline
22       & 12688   & 2 & 2696  &  1   & 923  & 56 & 6536\\
\hline
23       & 18570    & 10 & 3378  &     & 1412 & & 9080\\
\hline
24       &  25954    & 16 & 4204   &   14   & 2708  & 75 & 13046\\
\hline
25       & 44660   & 58  & 6940  &       & 3238 & & 26674\\
\hline
26       & 81579   &  146 &  10266  &  17   & 4343  & 742 & 50698\\
\hline
27       & 94158   & 252 & 11072  &      & 4576 & & 53522\\
\hline
28       & 132114   & 314 & 15864   &   16  & 5828  & 580 & 82028\\
\hline
29       & 159868   & 716 &  19754 &      & 7342 & & 106586\\
\hline
30       & 194674   & 1198 & 24096  &  8  & 10630  & 462 & 127450\\
\hline
31       & 247692   & 1314 & 29790  &      & 12124 & & 163450\\
\hline
32      &  404009  & 1982  &  38995  & 55 \hfill(15)  &  18441  & 7013 & 286060\\
\hline
33       & 356618   & 2822 & 40644  &       & 19194 & & 221900\\
\hline
34       & 480195  & 2706 & 52546  &   5    & 24711  & 5656 & 304066\\
\hline
35       & 461900   & 3452 & 54328  &        &30328  &  & 263654\\
\hline
36       & 609624   & 4342 &  66737  &  14    & 42359  & 5637 & 348343\\
\hline
37       & 678422   & 4484 & 66594  &      & 46796 &  & 374642\\
\hline
38       & 928441   & 5328  & 87391   &    11   & 61830  & 9793 & 536413\\
\hline
39       & 930304   & 5370 & 80096  &     & 60042 & & 503650\\
\hline
40       & 1244522   & 6412 & 102038   &   32  & 75913 & 12104 & 692657\\
\hline
41       &  1273738  & 6612  & 94144  &      & 76780 & & 668724\\
\hline
42       &  1739772  &7870  &  118548  &    14    & 98608  & 15092 & 965212\\
\hline
43       &  1849074  & 8914 & 106920  &       & 93256  & & 1012696\\
\hline
44       &  2701280  & 11736 & 141762   &   16    & 117140 & 35012 & 1583810\\
\hline
45       & 2644266   & 14948  &  118310 &       & 104706 & & 1462180\\
\hline
46       & 3491780   & 16876 &  156358  &  14     & 132576 & 34902 & 1933530\\
\hline
47       &  3686966  & 20004 & 126586  &        & 126006 & & 2060898\\
\hline
48       & 4911214   & 24830 &   167362   &  104  & 167322 & 59465 & 2895896\\
\hline
49       &  4790868  & 27420 &  128900 &       & 150198  & 64 & 2792242\\
\hline
50       &  6426642  & 34916  & 170856  &  12    & 184304  & 69244 & 4022222\\
\hline
51       &  6253444  & 40434 &  125366 &       &168414  & & 3932782 \\
\hline
52       &  8174653  & 47116 &  177956  &  41  & 211246  & 115157 & 5384015\\
\hline
53       & 7497364   & 56266 & 119424  &       & 178664 & &  5055996\\
\hline
\end{tabular}
\caption{MUCDs of degree 7}\label{tab71}
}
\end{table}

\begin{table}
{\small 
\begin{tabular}{|l|l|l|l|l|l|l|l|l|l|}
\hline
Size  & Total  & Connected & Normal &  \stackanchor{Self-dual}{(Symmetric)}  & Non-ample & Reducible&  Cop  \\
\hline
54       & 9180598   & 67628 & 162730   &   26    & 209307 & 119959 & 6487399\\
\hline
55       & 8270608   & 72728 &  108716 &       & 173954 & & 6107600\\
\hline
56       & 11160909   & 87290 & 161446   &   97    & 224096 & 223986 & 8766421\\
\hline
57       & 8540924   & 94064  & 95220  &       & 165036 & & 6742858\\
\hline
58       & 10269782   & 97952 &  134214 &  10     & 204109 & 178171 & 8360889\\
\hline	
59       &  7723932   & 94522 & 83966  &      & 149322 &     & 6345838	\\
\hline
60      & 9606176   & 92548 & 120399  &  52     & 202072 & 214186 & 8260766\\
\hline
61      & 6518148   & 77586 & 68314  &        & 131766 & & 5655112\\
\hline
62      &  7839946  & 69514 & 98479   &  22     & 159649 & 157801 & 6963611\\
\hline
63      & 5191166   & 55636 & 55530  &       & 108204 & 32 &  4642368\\
\hline
64      & 7728718   & 54052 &  85340  &  254\hfill (11)    & 173910   &  260912 & 7162308\\
\hline
65      & 3436076   & 38238 & 42744  &      & 93546 & & 3090834\\
\hline
66      & 3750621   & 34346 &  60993  &  39    & 112105 & 85176 & 3408640\\
\hline
67      &  2034070  & 29028 & 32490  &      & 60780 & & 1836672\\
\hline
68      & 2440206   & 26152 & 49547  &  42    & 97782 & 78262 & 2221040\\
\hline
69      &  1152526  & 20140 &  24736 &      & 47074 & & 1038388\\
\hline
70      & 1351871   & 19750 &  35886  &   11   & 65862 & 32445 & 1228087\\
\hline
71      &  671796  & 14742 &  17262 &      & 26368 & & 616530\\
\hline
72      &  808375  & 12776 & 24520  &   49   & 53157  & 25136 & 732188\\
\hline
73      & 357970  & 9872 &  10936 &      & 21338 & & 323602\\
\hline
74      & 405334   & 7714 & 15495   &  12  & 24711 & 9079 & 370471\\
\hline
75      & 186106   & 6120 & 7374  &     &  9364& & 171590\\
\hline
76      & 244369   & 4848 &  12120 &  7   & 16441 & 8798 & 223662\\
\hline
77      & 101268   & 3966 &  4818 &     & 5286 & & 94074\\
\hline
78      & 116958   & 3086 & 6400  &  2   & 6592 & 2562 & 108916\\
\hline
79      &  48120  & 2456 & 2792  &     & 2274 & & 45170\\
\hline
80      & 56464   & 1816 & 3719  &  2   & 3607 & 1294 & 52459\\
\hline
81      & 23402  & 1720 & 1490  &    & 1396 & 4 & 21864\\
\hline
82      & 25154   & 1208 & 1864  &     & 1506 & 350 & 23480\\
\hline
83      & 11344   & 1146  & 810  &    & 456 & & 10806\\
\hline
84      & 14503   & 938  & 1271  &  7  & 686 & 399 & 13799\\
\hline
85      & 6108   & 1020 & 370  &  &  254& & 5834\\
\hline
86      & 4273  & 552  &  506 & 1   & 222 & 49 & 4049\\
\hline
87      & 2066  & 506  & 226  &   & 96 & & 1970\\
\hline
88      & 2038   & 308  & 220  &   & 46 & 18 & 1992\\
\hline
89      & 1248   & 368 & 106   &   & 12 & & 1236\\
\hline
90      & 647   & 154 &  75 &  1 & 24 & 7 & 623\\
\hline
91      & 214   & 66 & 22  &   & 4 & & 210\\
\hline
92      & 274   & 98  & 46  &   &  & & 274\\
\hline
93      & 106   & 66  &  6 &   & 4 & & 102\\
\hline
94      & 76   & 36  &  10  &   &  & & 76\\
\hline
95      & 18   & 10  &  2 &   &  & & 18\\
\hline
96      & 36   & 30  &  8 &   &  & & 36\\
\hline
97      & 16   & 14 & 4  &   &  & & 16\\
\hline
98      & 4  & 4  &   &   &  & & 4\\
\hline
100    & 2  & 2  &  2 &   &  & & 2\\

\hline
\end{tabular}
\caption{MUCDs of degree 7}\label{tab72}
}
\end{table}

\subsection{Connectivity and Peak-Pit domains}\label{bruhat}
In this section we will consider several properties of  a UCD which are directly connected to the view of a CD as a subset of the permutohedron.

At least since the 1960's it has been common to consider \emph{connected} CDs, ie. a CD which induces a connected subgraph of the permutohedron. One attractive property of 
such domains is that it is possible to move between any two linear orders in the domain in step which only differ by an inversion. This can be interpreted as saying that the set of opinions is in some sense a continuum.     In the column labelled Connected we display the number of connected MUCDs of each size.    Here two things stand out in the data. First, the majority of all MUCDs are not connected. For small sizes, relative to $n$,  this is automatic but as we can see it seems to be the case for most sizes.    Secondly, up to $n=7$ the maximum MUCD is always connected.   We believe that the first of these properties holds more generally:
\begin{conjecture}
	A.a.s.  MUCDs are not connected.
\end{conjecture}
\begin{question}
	Are there always exactly 2 non-isomorphic  connected MUCDs of size ${n \choose 2}+1$?
\end{question}

In \cite{puppe2022maximal} Puppe and Slinko conjectured that a MUCD is connected if and only if it is a \emph{peak-pit domain}.  Peak-pit domains  stem from the early works of Black and Arrow  on single-peaked domains  and are defined as CD which on every triple either satisfy a condition of either the form $x$N1 or $x$N3, for some $x$ in the triple.  We have tested this conjecture on our data.
\begin{observation}
	For degrees $n\leq7$ a MUCD is connected if and only if it is a peak-pit domain.
\end{observation}

\subsection{Normal, symmetric, and self-dual MUCDS}
Two further classes of often-studied CDs are the \emph{normal}  and the \emph{symmetric} CDs.  The terminology in the literature varies a bit here but we 
will say that a CD is normal if it isomorphic to a CD which contains both the standard order $\alpha$ and the reverse order $u$. This is sometimes instead called normalisable, with normal then meaning that the CD actually contains   $\alpha$ and $u$, and sometimes called being of maximal width.    Being symmetric on the other hand means that for every order $\beta$ in the domain $C$  the reversed order $u\beta$ also belongs to $C$.     In the columns labelled Normal  and (Symmetric) we give the number of normal  and symmetric MUCDs. As we can see, the number of normal MUCDs is substantially smaller than the total  number, and we believe that this patterns will continue.
\begin{conjecture}
	A.a.s.  MUCDs are not normal.
\end{conjecture}
We also note that for degree $n\leq 7$  the maximum MUCD is always normal. However, in \cite{LMR23} the maximum MUCD of degree 8 was found and it is not normal.  Here one may ask if normality implies a strong restriction on the size of a MUCD.
\begin{question} 
	Is the maximum size of a normal MUCD $o(F(n))$?
\end{question}

The symmetric MUCDs form a subfamily of the \emph{self-dual} MUCDs.
\begin{definition}
	A MUCD  $C$ is self-dual  if the dual MCD   $uC$ is isomorphic to $C$
\end{definition} 
Note that if we demand that the dual is equal, instead of isomorphic, to the original MUCD the  we  get a symmetric MUCD.  The number of self-dual MUCD are given in the column labelled Self-dual. Here we see that  while the number of possible sizes for a self-dual MUCD  is much larger than for the symmetric ones  the total number of self-dual MUCDs is still a small proportion of the total.  However, we also note that for odd $n$  the maximum MUCD are all self-dual  for  $n\leq 7$.
\begin{question}
	Are maximum MUCDs self-dual for odd $n$?  If not, which is the smallest $n$ for which the maximum MUCD is not self-dual?
\end{question}

A second observation is that for odd $n$  we have only seen self-dual MUCDs with even size.
\begin{question}
	Do all self-dual MUCDs have even size if $n$ is odd?
\end{question}

Both normality and being symmetric can be seen as  properties of the intersection between a domain $C$  and the dual domain $uC$. A domain is normal if the intersection is non-empty  and symmetric if the intersection is equal to the entire domain. Note that, since the $\beta$ is never equal to $u\beta$, the intersection will always have even size.  In Tables \ref{int45}, \ref{int6}, \ref{int71}, and \ref{int72}  we give the number of MUCDs of each degree and size with a given size for the intersection.  As one might expect the two most common intersection sizes are 0 and 2. Having intersection size 4 is possible for many sizes and from degree 7 is no longer connected to having  an even domain size, as sizes 49, 63 and 81 show.  Also note that for domains of size 8 the proportion of symmetric domains increases with $n$ and for $n=7$ all MUCDs of size 8 are symmetric.
\begin{question}
	Are all MUCDs of size 8 symmetric  for $n\geq 7$?
\end{question}

Both normality and being symmetric can be seen as  properties of the intersection between a domain $C$  and the dual domain $uC$. A domain is normal if the intersection is non-empty  and symmetric if the intersection is equal to the entire domain. Note that, since the $\beta$ is never equal to $u\beta$, the intersection will always have even size.  In Tables \ref{int45}, \ref{int6}, \ref{int71}, and \ref{int72}  we give the number of MUCDs of each degree and size with a given size for the intersection.  As one might expect the two most common intersection sizes are 0 and 2. Having intersection size 4 is possible for many sizes and from degree 7 is no longer connected to having  an even domain size, as sizes 49, 63 and 81 show.  Also note that for domains of size 8 the proportion of symmetric domains increases with $n$ and for $n=7$ all MUCDs of size 8 are symmetric.
\begin{question}
	Are all MUCDs of size 8 symmetric  for $n\geq 7$?
\end{question}

Note that up to $n=7$ the intersections always have a power of 2 as its size. This is true in general as we will now show. Additionally, in \cite{KS22}  the problem of determining the possible sizes for symmetric MUCDs was raised, after noting that all known constructions give, all, powers of 2 as  size.  This question was in fact implicitly solved already in \cite{danilov2013}  and it also follows from out theorem.
\begin{theorem}
	Let $I$ denote the intersection of a MUCD  $C$ and its dual.  Then the size of $I$ is $2^{k}$, for an integer $k$,  and if $C$ contains the reversed order $u$  then 
	$I$ induces a Boolean sublattice of the weak Bruhat order.
\end{theorem}
\begin{proof}
		First we note that $I$ is by definition the largest symmetric, meaning equal to its dual, subset of $C$, and we can assume that it contains $\alpha$ and $u$.  Now, 
		as shown in \cite{danilov2013}   $C$  induces a distributive sublattice of the weak Bruhat order.   Taking two elements $\sigma, \tau \in I$   it follows that 
		$(\sigma \wedge \tau)^\circ=\sigma^\circ\vee\tau^\circ$, where the $\circ$ denotes the reversed order $\tau^\circ=u\tau$. That is, the reverse of the meet of any pair of 
		orders in $I$  is the join of their reverses. So if we add the meet of any  two orders from $I$ and the join of their reverses we get a symmetric set.  But since $I$ is the 
		maximum symmetric subset it must be closed under taking meets and joins.
		
		Next let us note that in this lattice the meet and join  of  an order $\beta$ and its reverse $\beta^\circ$ are $\alpha$ and $u$ respectively. This follows since the set of 
		inversions of $\beta^\circ$ is the complement of the set of inversion of $\beta$.   This means that the reverse $\beta^\circ$ satisfies the conditions for being a 
		\emph{complement} of $\beta$ in the lattice-theoretic sense. Since the lattice is distributive it also follows that  $\beta^\circ$ is the unique complement for $\beta$. 

		So our domain $C$  induces a finite, distributive, complemented lattice  and by  e.g. Theorem 16, Chapter 10, in \cite{B48}  all such lattices are isomorphic to a Boolean 
		lattice, and hence have size $2^k$ for some integer $k\geq 0$.
\end{proof}
By our proof the intersection sets $I$ are Condorcet domains which are closed under meets and joins in the Bruhat order, however they are typically not maximal  Condorcet domains.

\begin{table}[hbt!]
\begin{tabular}{|l|l|l|l|l|l|l|l|}
\hline
Degree &Size  &   &&&&\\
  & &  0 & 2 & 4 &  8 & 16  \\
\hline
4 & 4 & & & 1  & &\\
\hline
4 & 7 & & 4  &&&\\
\hline
4 & 8 &  9 & 8  &  6  & 2  &\\
\hline
4 & 9 & & 1 & &&\\
\hline
\hline

&&&&&&\\
\hline
5 & 4 & & & 2  && \\
\hline
5 & 8 &  4 &  6 & & 2 &\\
\hline
5 & 11 &  10 & 18  & &&\\
\hline
5 & 12 &  9 & 32  &&&\\
\hline
5 & 13 &  20  &  32  &&&\\
\hline
5 & 14 &  161 & 98  & 20  &&\\
\hline
5 & 15 &  154  &  58  &&&\\
\hline
5 & 16 &  432  &  78  &  44  & 16  & 3  \\
\hline
5 & 17 &  72  &  34  &&&\\
\hline
5 & 18 &  24  & 14  &  5  & &\\
\hline
5 & 19 &  6  &  6  & &&\\
\hline
5 & 20 & &  2  & &&\\
\hline
\end{tabular}
\caption{Size of the intersection between C and the reverse of C}\label{int45}
\end{table}

\begin{table}
\begin{tabular}{|l|l|l|l|l|l|l|l|}
\hline
Degree &Size  & &&&&&  \\
  & &  0 & 2 & 4 &  8 & 16 & 32 \\
\hline
6 & 4 &   &   & 8 & &&\\
\hline
6 & 8 &  4 &   &  &  7  &&\\
\hline
6 & 9&  8 & 18 & &&&\\
\hline
6 & 10 & 18 & 28 & &&& \\
\hline
6 & 11&  2 &  6  &  &&&\\
\hline
6 & 12&  7 &  4 &  &&&\\
\hline
6 & 13&  68 & 38 &  &&&\\
\hline
6 & 14&  48 &  32 &  &&&\\
\hline
6 & 15&  32 & 34  & &&&\\
\hline
6 & 16&  790 &  202  &  32  & 6 & 6 & \\
\hline
6 & 17&  564 &  244 &  &&&\\
\hline
6 & 18&  528 & 280  &  &&&\\
\hline
6 & 19&  862 &  537  &  &&&\\
\hline
6 & 20&  1070 &  664 &  &&&\\
\hline
6 & 21&  1448 &   708 & &&&\\
\hline
6 & 22&  3878 &  1086  & 108 &&& \\
\hline
6 & 23&  3608 &  1378 & & &&\\
\hline
6 & 24&  6767 &   2166 & 184  &&&\\
\hline
6 & 25&  8268 &   1624 &  &&&\\
\hline
6 & 26&   14127 &   2310 & 192  &&&\\
\hline
6 & 27&  15381 &  1756  & & &&\\
\hline
6 & 28&  29608 &  2416  & 620  & 64 && \\
\hline
6 & 29&  23693 &  1760   & &&&\\
\hline
6 & 30&  29021 &  1941  & 348  &&&\\
\hline
6 & 31&  21162 &  1381 & &&&\\
\hline
6 & 32&  36699 &  1450 & 536 &  163 &  40 & 6  \\
\hline
6 & 33&  11347 &   821 &  &&&\\
\hline
6 & 34&  10479 &  871 & 204 &&& \\
\hline
6 & 35&  4103 &  532 & &&&\\
\hline
6 & 36&  3262 &  350 & 92  & 16 && \\
\hline
6 & 37&  1120 &  177  & &&&\\
\hline
6 & 38&  1016 &  248 & 36  &&&\\
\hline
6 & 39&  296 &   70 &  &&& \\
\hline
6 & 40&  151 &  33 & 8  &&&\\
\hline 
6 & 41&   34 &  16  & &&& \\
\hline
6 & 42&   42 &  15  & &&& \\
\hline
6 & 43&   5 &  2 & &&& \\
\hline
6 & 44&   2 &  2 & &&&\\
\hline
6 & 45&   0 & 1 & &&&\\
\hline
\end{tabular}
\caption{Size of the intersection between C and the reverse of C }\label{int6}
\end{table}

\begin{table}
{\small 
\begin{tabular}{|l|l|l|l|l|l|l|l|l|}
\hline
Size  &  &&&&&& \\
         & 0 & 2 & 4 & 8 & 16 &32 & 64  \\
\hline
4       &     &  & 46 &   &  & &\\
\hline
8       &     &  & & 44 &  & &\\
\hline
9       & 24   &  & &  &  & &\\
\hline
10      & 84   & 186  & &  &  & &\\
\hline
11       & 68   & 120 & &  &  & &\\
\hline
12       & 63   &  84 & &  &  & &\\
\hline
13       &  104  & 72 & &  &  & &\\
\hline
14       & 174   & 110  & &  &  & &\\
\hline
15       &  360  & 188  & &  &  & &\\
\hline
16       & 1174   & 410 & & 18 & 24 &  & \\
\hline
17       & 1688   & 490 & &  &  & &\\
\hline
18       & 3641   & 668 & 126 &  &  & &\\
\hline
19       &  8746  & 1248 & &  &  & &\\
\hline
20       & 10320   & 1348 & 196 &  &  & &\\
\hline
21       & 5766   & 1274 & &  &  & &\\
\hline
22       & 9989   & 2654 & 42 &  &  & &\\
\hline
23       & 15192  & 3378 & &  &  & &\\
\hline
24       & 21750   & 4176 & 28 &  &  & &\\
\hline
25       & 37720   & 6940 & &  &  & &\\
\hline
26       & 71313   & 10000  & 266 &  &  & &\\
\hline
27       & 83086   & 11072 & &  &  & &\\
\hline
28       & 116250   & 15629 & 224 & 20  &  & &\\
\hline
29       & 140114   & 19754 & &  &  & &\\
\hline
30       & 170578   & 23858 & 238 &  &  & &\\
\hline
31       & 217902   & 29790 & &  &  & &\\
\hline
32      &  365014  & 37322  & 1472 & 152 & 34  & 15 &\\
\hline
33       & 315974   & 40644 & &  &  & &\\
\hline
34       & 427649   & 50838 & 1708 &  &  & &\\
\hline
35       & 407572   & 54328 & &  &  & &\\
\hline
36       & 542887   & 64778 & 1954 & 5 &  & &\\
\hline
37       & 611828   & 66594  & &  &  & &\\
\hline
38       & 841050   &  83632 & 3759&  &  & &\\
\hline
39       & 850218   & 80086 & &  &  & &\\
\hline
40       & 1142484   & 97396 & 4642 &  &  & &\\
\hline
41       &  1179594  & 94144  & &  &  & &\\
\hline
42       &  1621224  & 113592 & 4956 &  &  & &\\
\hline
43       &  1742104  & 106970 & &  &  & &\\
\hline
44       &  2559518  & 133512 & 7854 & 396 &  & &\\
\hline
45       &  2525956  & 118310  & &  &  & &\\
\hline
46       &  3335422  & 146712 & 9646 &  &  & &\\
\hline
47       &  3560380  & 126586 & &  &  & &\\
\hline
48       &  4743852  & 154626 & 12064 & 672 &  & &\\
\hline
49       &  4661968  & 128836 & 64&  &  & &\\
\hline
50       &  6255786  & 159488  & 11368 &  &  & &\\
\hline
51       &  6128078  & 125366 & &  &  & &\\
\hline
52       &  7996697  & 160682 & 16570& 704 &  & &\\
\hline
53       &  7377940  & 119424 & &  &  & &\\
\hline
\end{tabular}
\caption{Size of the intersection between C and the reverse of C for degree 7}\label{int71}
}
\end{table}

\begin{table}
{\small 
\begin{tabular}{|l|l|l|l|l|l|l|l|l|}
\hline
Size  &  &&&&&& \\
         & 0 & 2 & 4 & 8 & 16 &32 & 64  \\
\hline
54       &  9017868  & 150438 & 12292 &  &  & &\\
\hline
55       &  8161892  & 108716 & &  &  & &\\
\hline
56       &  10999463  & 140374 & 18444 & 2444 & 184 & &\\
\hline
57       &  8445704  & 95220  & &  &  & &\\
\hline
58       &  10135568  & 121894 & 12320 &  &  & &\\
\hline
59       &  7639966  & 83966 & &  &  & &\\
\hline
60       &  9485777  & 104754 & 14369 & 1276 &  & &\\
\hline
61      & 6449834   & 68314 & &  &  & &\\
\hline
62      & 7741467   & 88812 & 9667 &  &  & &\\
\hline
63      & 5136516   & 55498 & 32 &  &  & &\\
\hline
64      & 7646022   & 71024 & 11378& 2314 & 515 & 98 & 11\\
\hline
65      & 3393332   & 42744 & &  &  & &\\
\hline
66      & 3689628   & 55246 & 5747 &  &  & &\\
\hline
67      & 2001580   & 32490 & &  &  & &\\
\hline
68      & 2391499   & 42232 & 6567& 748 &  & &\\
\hline
69      & 1127790   & 24736 & &  &  & &\\
\hline
70      & 1315985   & 32162 & 3724&  &  & &\\
\hline
71      & 654534   & 17262 & &  &  & &\\
\hline
72      & 783855   & 21414 & 2680 & 380  & 46 & &\\
\hline
73      & 347034   & 10936 & &  &  & &\\
\hline
74      & 389839   & 14256 & 1239&  &  & &\\
\hline
75      & 178732   & 7374 & &  &  & &\\
\hline
76      & 232249   & 10174 & 1814 & 132 &  & &\\
\hline
77      & 96450   & 4818 & &  &  & &\\
\hline
78      & 110558   & 5910 & 490&  &  & &\\
\hline
79      & 45328   & 2792 & &  &  & &\\
\hline
80      & 52745   & 3446 & 245 & 28 &  & &\\
\hline
81      & 21912   & 1486 & 4 &  &  & &\\
\hline
82      & 23290   & 1752  & 112 &  &  & &\\
\hline
83      & 10534   & 810  & &  &  & &\\
\hline
84      & 13232   & 1166  & 105 &  &  & &\\
\hline
85      & 5738   & 370  & &  &  & &\\
\hline
86      & 3767   & 492  & 14 &  &  & &\\
\hline
87      & 1840   & 226  & &  &  & &\\
\hline
88      & 1818   & 212  & 8 &  &  & &\\
\hline
89      & 1142   & 106  & &  &  & &\\
\hline
90      & 572   & 68  & 7 &  &  & &\\
\hline
91      & 192   & 22  & &  &  & &\\
\hline
92      & 228   & 46  & &  &  & &\\
\hline
93      & 100   & 6  & &  &  & &\\
\hline
94      & 66   &10   & &  &  & &\\
\hline
95      & 16   & 2  & &  &  & &\\
\hline
96      & 28   & 8  & &  &  & &\\
\hline
97      & 12   & 4  & &  &  & &\\
\hline
98      & 4   &   & &  &  & &\\
\hline
100    &    & 2  & &  &  & &\\

\hline
\end{tabular}
\caption{Size of the intersection between C and the reverse of C for degree 7}\label{int72}
}
\end{table}

\section{Relation to other domain types}\label{other}
The main motivation for studying Condorcet domains has been to better understand majority voting, as in Condorcet's original work. However, today domains of linear 
orders are studied much more broadly, both in connection with other classical voting  systems  and regarding where well-behaved voting systems or choice rules can be constructed.
The work of Dasgupta and Maskin  \cite{DM08} shows that Condorcet domains are the largest domains where any voting system satisfies a specific list of axioms for a well-behaved voting system.  So, in this broader context Condorcet domains stand out  in this sense,  but many authors focus on weaker axioms and we will here briefly comment on how the Condorcet domains for small $n$  relate to two such lines of investigation.

Recall that a voting systems is \emph{strategy-proof}, or non-manipulable,  if the best option for each voter is to present a ranking which agrees with their actual preferences.  The classical   Gibbard-Satterthwaite theorem \cite{Gi73,Sa75} states that if the domain consists of all, unrestricted, linear orders  then the only strategy-proof deterministic voting system is dictatorial, i.e. the outcome depends only on one voter.  On the other hand, majority voting  on Condorcet domains  is not only strategy-proof but even proof against strategic voting by coalitions of voters, see Lemma 10.3 of \cite{Mou88}.  A number of papers have investigated either how much a domain can be restricted while retaining the conclusion from the Gibbard-Satterthwaite theorem  or how large a domain can be while allowing non-dictatorial  choice functions.  In \cite{ACS03} Aswal, Chatterji and Sen  introduced the \emph{unique seconds property}, abbreviated USP,   and showed that any domain with the USP has a non-trivial strategy-proof choice function.  A domain has the USP  if there exists a pair of alternatives A and B such that whenever A is ranked first in a linear order B is ranked second. The property has turned out to be quite fruitful and recently \cite{CZ23} showed that in a certain well-connected class of domains the USP is in fact equivalent to the existence of non-trivial strategy-proof choice functions.  

Given that we already know that CDs are strongly strategy-proof we may ask how they fit in the wider landscape of strategy-proof domains, and in particular if they have the USP.   It turns out that  among the CDs for small $n$ many do in fact have the USP, but far from all do.  In Table \ref{tabOD} we show the number of CDs with the USP for degree 4 and 5.  Since the USP is not invariant under reversal of orders it can happen that a CD does not have the USP but its dual does, and this is quite common.  Therefore we also show the number of domains such that neither the domain nor its dual has the USP.   These provide examples of strategy-proof domains which are not covered by USP condition for strategy-proofness, and as we find such examples close to the maximum size for CDs of these degrees.  If we simply demand that the domain does not have the USP then one of the two maximum CDs for $n=5$ is also an example\footnote{Data for $n=6, 7$ can be found in the online appendix.}. 

Another line of work, which intertwines with strategy-proofness, concerns generalisations of Black's single-peaked MUCD.  For each $n$ there is up to isomorphism one Black's single-peaked domain, of size $2^{n-1}$. This is a particularly well-behaved MUCD arising from preferences based  on positions on a linear axis, which can be characterised in various ways~\cite{BH11,PUPPE2018}. This MUCD was first generalised by Arrow into what is now known as Arrow's single-peaked domains. These domains are also MUCDs  but unlike Black's version  there are several non-isomorphic examples for each $n$.  Put briefly a MUCD is Arrow's single-peaked if every triple $(i,j,k)$ satisfies a never condition of the form $x$N3, where $x$ is a member of the triple.  In Slinko's study of these domains~\cite{SLINKO2019166} he enumerated them for $n=4, 5$  and from our data we can extend this:
\begin{observation}
	The number of non-isomorphic Arrow's single-peaked MUCDs for $n=4,\ldots,7$ is $2, 6 ,40, 560$.
\end{observation}
Stepping outside the class of Condorcet domains Demange~\cite{dem82} defined the class of domains which are \emph{single-peaked on a tree}. Here a domain $D$ on $X_n$ is said to be single-peaked on a tree $T$ with $n$ vertices if we can label the vertices in $T$ with the alternatives from $X_n$ so that the restriction of $D$ to the labels of any maximal path in $T$ is a Black's single peaked domain.  These domain are often not CDs  but they have the weaker property of guaranteeing that pairwise majorities selects a single winner, while there may be cycles among lower-ranked alternatives.   For Black's single-peaked domain Moulin~\cite{Mou80} has identified all strategy-proof choice functions and Danilov~\cite{DANILOV1994} extended this to domains which are single-peaked on a tree. In particular these domains always have a strategy-proof choice function and so ties in with the already mentioned works on strategy-proofness.
Recently these domains have also been the focus for development of efficient algorithms, see \cite{peters22} and references therein.

Here it becomes natural to ask how common it is for Condorcet domains  to be single-peaked on a tree and it turns out to be a rare property for MUCDs.   In Table \ref{tabOD} we give both the total number of MUCDs which are single-peaked on a tree and those which are single-peaked on a star\footnote{Data for $n=6, 7$ can be found in the online appendix.}. 

\begin{table}
\begin{tabular}{|l|l|l|l|l|l|l|l|l|l|l|}
\hline
Degree &Size  & Total  & USP  & NUSPD &   & SPT & Star &  \\
 & &  &   &   &    &   && \\
\hline
4 & 4 &  1 &  1 &     &  &  & &  \\
\hline
4 & 7 &  4  &  2   &   &   &  &  &  \\
\hline
4 & 8 & 25  &  16  &  3  & & 3 & 2 &  \\
\hline
4 & 9 &  1  &  1 &   &   &   &  & \\
\hline
\hline

&&&&&&\\
\hline
5 & 4 & 2   &   2 &      &       &     &  &  \\
\hline
5 & 8 &  12 &  10  &     &   &  1   &  1 &   \\
\hline
5 & 11&  28   &  12   &   4  &   &  &  &   \\
\hline
5 & 12 &  41  &  19  &  5    &    &    &  &   \\
\hline
5 & 13 & 52    & 33  &   2   &   &  &  &   \\
\hline
5 & 14 &  279  &   155  &   46    &    & 3   &  1  &    \\
\hline
5 & 15 &  212   &  96   &  44  &    & &   & \\
\hline
5 & 16 &  573   &  380  &  49    &    &  18    &  10   &   \\
\hline
5 & 17 &  106   &  87   &  2   &    &  & &   \\
\hline
5 & 18 &  43   &   31  &   4  &   &  &  & \\
\hline
5 & 19 & 12    &  9  &     &   &  & &  \\
\hline
5 & 20 &  2  &  1  &    &  & & &  \\
\hline
\end{tabular}
\caption{The column NUSPD counts MUCDs such that neither the domain nor its dual has the USP. The column SPT counts MUCDs which are single-peaked for a tree.}\label{tabOD}
\end{table}

\section*{Acknowledgements}
This research was conducted using the resources of High Performance Computing Center North (HPC2N).  We would like to thank Alexander Karpov for useful comments on the first version 
of the manuscript.


\newpage
\section*{Online Appendix}
Here we present additional data for Section 5 of our paper.

\begin{table}
\begin{tabular}{|l|l|l|l|l|l|l|l|}
\hline
Size  & Total &    & USP& NUSPD  &   &  SPT &  Star \\
\hline
 4 & 8  && 8    &    &   &   &  \\
\hline
 8 & 11  && 10   &     &     & 1 & 1 \\
\hline
 9& 26  &&  20   &   && & \\
\hline
 10 & 46 && 40    &  && 1 & \\
\hline
 11& 8  &    & 7    &   &  & & \\
\hline
 12& 11  &   &  8   &  1  &  & & \\
\hline
 13& 106 &   &  67   & 4   &  &  &  \\
\hline
 14& 80  &    &  61   &   &  & &   \\
\hline
 15& 66  &     &  53   &   &  & &   \\
\hline
 16& 1036  &     & 719   & 64      &      &  11  &  7  \\
\hline
 17& 808  &     &  413   &  140  &  & &   \\
\hline
 18& 808  &     &   379    &  128  &  & &    \\
\hline
 19& 1399  &   &   670   &  207    &  & &     \\
\hline
 20& 1734 &      &  839    &  258    &  & &     \\
\hline
 21& 2156  &    &    1118  &   289    &  & &    \\
\hline
 22& 5072 &      & 2561    &  876   &   & 8   &  2 \\
\hline2 
 23& 4986  &    &   2677    & 682    &     & &  \\
\hline
 24& 8617 &     &  4565      & 1386    & &  14   & 3 \\
\hline
 25& 9892  &    &  4804     &  2164    &  & & \\
\hline
 26& 16629 &      &  8823    & 3129     &  &  29 &  9\\
\hline
 27& 17137 &      &  8460     &  3717     &    & &      \\
\hline
 28& 32708  &       &   17428     &  5864     &   &  100   &   30  \\
\hline
 29& 25453  &      &  13241     &  4709    &    & &      \\
\hline
 30& 31310  &       &  17213     & 4752  &   &   44  &  10   \\
\hline
 31& 22543 &       &  12761     &  3498  &   & &     \\
\hline
 32& 38894 &       & 26102       & 3242     &      &288    &    126   \\
\hline
 33& 12168 &       & 8872        &  710   &    & &      \\
\hline
 34& 11554  &      & 8385       &  788  &   &   8  &   6  \\
\hline
 35& 4635  &      &  3429    &  282 &    & &    \\
\hline
 36& 3720  &      &  2698    & 270     &     &10 &  8   \\
\hline
 37& 1297  &      & 897      & 73   &   & &    \\
\hline
 38& 1300 &     &   930    &  90   &  &  1  &    1  \\
\hline
 39& 366  &     & 270    &  8  & &  &    \\
\hline
 40& 192 &     &  147    &  1  &   &   1 &   1  \\
\hline 
 41& 50   &      & 36  & 2 && &     \\
\hline
 42& 57   &     &  36  & 6  &&  &    \\
\hline 
 43& 7   &    &  4   &  1  &&  & \\
\hline
 44&  4  &    &  3  &   && & \\
\hline
45&  1  &    &  1   &    && & \\
\hline

\end{tabular}
\caption{MUCDs of degree 6 }
\end{table}

'
\begin{table}
\begin{tabular}{|l|l|l|l|l|l|l|l|l|}
\hline
Size  & Total &    & USP& NUSPD  &   &  SPT &  Star \\
\hline
4       &  46  &&   46     &      &&    &      \\
\hline
8       & 44   &&   44    &      &&    &      \\
\hline
9       & 24   &&  21     &      &&    &      \\
\hline
10      & 270 &&  204     &      &&   1 &      \\
\hline
11       & 188  &&  158     &      &&    &      \\
\hline
12       & 147   &&   132    &  1    &&   2 &      \\
\hline
13       & 176 &&  152     &      &&    &      \\
\hline
14       & 284  &&   225    &      &&  3  &      \\
\hline
15       & 548  &&  356     &  28    &&    &      \\
\hline
16       & 1626  &&  1147     & 54     &&  8  &   5   \\
\hline
17       & 2178   &&   1715    & 26     &&    &      \\
\hline
18       & 4435   &&  3579     & 18     &&  5  &   3   \\
\hline
19       & 9994   &&  7595     &  204    &&    &      \\
\hline
20       & 11864   && 9416      &  308    && 39   &   3   \\
\hline
21       & 7040   &&  5334     & 268     &&    &      \\
\hline
22       & 12685   &&  7880     &  1046    &&  3  &      \\
\hline
23       & 18570   &&  10502     &  2382    &&    &      \\
\hline
24       & 25954  &&  13337     &  4083    &&  3  &  1    \\
\hline
25       & 44660   && 21440      &  8310    &&    &      \\
\hline
26       & 81579   &&  41554     & 15134     &&  55  &   19   \\
\hline
27       & 94158  &&  48104     &  16490    &&    &      \\
\hline
28       & 132114  &&  70626    &  20418    &&  38  &   24   \\
\hline
29       & 159868  && 86275      &   24034   &&    &      \\
\hline
30       & 194674   &&  109741     & 24592     && 20   &    14  \\
\hline
31       & 247692   &&  135543     & 36242     &&    &      \\
\hline
32      &  404009  &&  232859     &  55893    &&  547 &   294   \\
\hline
33       & 356618   &&  187736     & 58636     &&    &      \\
\hline
34       & 480195   &&  241879     &  84736    && 141   &    80  \\
\hline
35       & 461900   &&  232880     &  77574    &&    &      \\
\hline
36       & 609624   &&  306168     & 104774     && 132   &  41    \\
\hline
37       & 678422   && 332654      & 126548     &&    &      \\
\hline
38       & 928441   &&  468659     &  169820    &&  201  &   48   \\
\hline
39       & 930304   &&  462464     &   177382   &&    &      \\
\hline
40       & 1244522  &&  622926     &  245049    &&  263  &    88  \\
\hline
41       & 1273738  && 635695      &  257562    &&    &      \\
\hline
42       & 1739772  &&  868016     &  357082    &&  423  &   117   \\
\hline
43       &  1849074 &&  891637     &  413856    &&    &      \\
\hline
44       &  2701280 &&  1341710    &  560391    &&  959  &   223   \\
\hline
45       &  2644266 &&   1274294    &  598920    &&    &      \\
\hline
46       &  3491780 && 1717713      & 766557     &&  754  &   196   \\
\hline
47       &  3686966 &&  1743701     &  901630    &&    &      \\
\hline
48       &  4911214 &&  2396695     &  1117169    &&   1547 &   433   \\
\hline
49       &  4790868 &&  2276175     &  1168420    &&    &      \\
\hline
50       &  6426642 &&  3155872     &  1445189    &&   1173 &   314   \\
\hline
51       &  6253444 &&  2975189    &  1527476    &&    &      \\
\hline
52       &  8174653 && 4077249    &  1795796    &&  2935  & 1067     \\
\hline
53       &  7497364 &&  3595819     &  1801242    &&    &      \\
\hline
\end{tabular}
\caption{MUCDs of degree 7}
\end{table}

\begin{table}
\begin{tabular}{|l|l|l|l|l|l|l|l|l|l|}
\hline
Size  & Total &    & USP& NUSPD  &   &  SPT &  Star \\
\hline
54       &  9180598  && 4590919   &  1975335    &&   2018 &    693  \\
\hline
55       &  8270608  &&  4045313     &   1893692   &&    &      \\
\hline
56       &  11160909 &&  5810955     &  2141477    &&  6714  &    2185  \\
\hline
57       &  8540924  && 4321308     &  1759986    &&    &      \\
\hline
58       & 10269782 &&   5566271    &  1761191    &&  2637 &   1345  \\
\hline	
59       & 7723932   &&  4138772     & 1393850      &&    &      \\
\hline
60      & 9606176  && 5482135     &   1410202   &&   5002 &   1953   \\
\hline
61      & 6518148  &&  3653765     &   1029668   &&    &      \\
\hline
62      & 7839946  &&  4675901     &  996210    &&  2821  &  1180   \\
\hline
63      & 5191166  && 3074071      & 697030     &&    &      \\
\hline
64      & 7728718  &&  5147199     &  683048    && 12934   &  5762    \\
\hline
65      & 3436076  &&   2292393   & 315782     &&    &      \\
\hline
66      & 3750621  &&  2596227     &  302273    &&  1238  &   870   \\
\hline
67      & 2034070  &&  1387336     &  177286    &&    &      \\
\hline
68      & 2440206  &&  1737293     &  182585    &&  1891  &   1209   \\
\hline
69      & 1152526  &&   793497   &  100214    &&    &      \\
\hline
70      & 1351871  &&  965314     &  99595    &&  492  &   356   \\
\hline
71      & 671796   &&  465452     & 58322     &&    &      \\
\hline
72      & 808375   &&  584378     &  55400    &&  805  &    523  \\
\hline
73      & 357970  && 248428       &  29754    &&    &      \\
\hline
74      & 405334  && 294744     &  25405    &&  142  &   108   \\
\hline
75      & 186106  &&  131071     &  13680    &&    &      \\
\hline
76      & 244369  &&  180642    &  12648    &&  272  &   180   \\
\hline
77      & 101268  &&  72371     &   6398   &&    &      \\
\hline
78      & 116958  &&  86951   &   5556   &&  34  &  25    \\
\hline
79      & 48120   &&   33469    &  3434    &&    &      \\
\hline
80      & 56464   &&   41434    &   2637   &&  70  &   52   \\
\hline
81      & 23402   && 16423      &   1498   &&    &      \\
\hline
82      & 25154   &&  17915     &   1568   &&   10 &  6    \\
\hline
83      & 11344   &&  7570     &    1058  &&    &      \\
\hline
84      & 14503   &&  10183     &   954   &&  12  &  8    \\
\hline
85      & 6108   &&   4061    &   516   &&    &      \\
\hline
86      & 4273  &&   3094    &   257   &&    &      \\
\hline
87      & 2066  &&  1419    &   164   &&    &      \\
\hline
88      & 2038  &&   1524    &   76   && 2   &   2   \\
\hline
89      & 1248  &&  865     &   98   &&    &      \\
\hline
90      & 647   &&   486    &   24   &&    &      \\
\hline
91      & 214   &&  170     &   14   &&    &      \\
\hline
92      & 274   && 209      &  16    &&    &      \\
\hline
93      & 106   &&   77    &  8    &&    &      \\
\hline
94      & 76   &&  59     &   6   &&    &      \\
\hline
95      & 18   &&    10   &  4    &&    &      \\
\hline
96      & 36   &&  27     &  2    &&    &      \\
\hline
97      & 16   &&  12     &      &&    &      \\
\hline
98      & 4  &&   2    &      &&    &      \\
\hline
100    & 2  &&   1    &      &&    &      \\

\hline
\end{tabular}
\caption{MUCDs of degree 7}
\end{table}

\end{document}